\documentclass[letterpaper, 10 pt, conference]{ieeeconf}  

\IEEEoverridecommandlockouts                              
\title{\LARGE \bf
Simultaneous improvement of control and estimation for battery management systems
}

\usepackage{graphicx}
\usepackage{booktabs}
\usepackage{caption}
\usepackage{subfigure}
\usepackage{amsmath} 
\usepackage{amsthm}
\let\labelindent\relax
\usepackage{enumitem}
\usepackage{amsfonts}
\usepackage{amssymb}  
\usepackage{mathtools} 
\usepackage{mathrsfs}
\usepackage{balance}
\usepackage{url}
\usepackage{hyperref}
\usepackage{algorithm}
\usepackage{algorithmicx}
\usepackage{algpseudocode}
\DeclareMathAlphabet{\mathpzc}{OT1}{pzc}{m}{it}
\usepackage{makecell}

\usepackage{soul}
\usepackage{multirow}

\theoremstyle{definition}

\newtheorem{corollary}{Corollary}
\newtheorem{lemma}{Lemma}

\newtheorem{remark}{Remark}

\newtheorem{problem}{Problem}

\newcommand{\E}{\mathbb{E\,}}
\newcommand{\Cov}{\textbf{Cov}}

\newcommand{\doi}[1]{\href{http://dx.doi.org/#1}{\normalsize{\textsc{doi:}}~\nolinkurl{#1}}}
\newcommand{\arxiv}[1]{\href{http://arxiv.org/abs/#1}{\normalsize{\textsc{arxiv:}}~\nolinkurl{#1}}}
\usepackage{color}

\usepackage{comment}

\renewcommand{\epsilon}{\varepsilon}

\usepackage[dvipsnames]{xcolor}

\newcounter {dagger} 
\author{Mohammad S. Ramadan$^{1} _\star$,\,Marfred Barrera$^{2} _\star$,\,Mihai Anitescu$^1$,\,Sylvia Herbert$^2$
\thanks{$1$ Mathematics and Computer Science Division, Argonne National Laboratory, Lemont, IL 60439, USA,  {\tt\footnotesize mramadan@anl.gov, anitescu@mcs.anl.gov.}}
\thanks{$2$ Mechanical and Aerospace Engineering at UC San Diego, La Jolla, CA 92093, USA,  {\tt\footnotesize mrbarrera@ucsd.edu, sherbert@ucsd.edu.}}
\thanks{
This material was based upon work supported by: the U.S. Department of Energy, Office of Science, Office of Advanced Scientific Computing Research (ASCR) under Contract DE-AC02-06CH11347; and ONR YIP (\#N00014-22-1-2292) and the UCSD JSOE Early Career Faculty Award. 
}%
\thanks{$\star$ Equal contribution}
}

\usepackage[noadjust]{cite}

\begin{document}
\maketitle

\begin{abstract}
Standard battery management systems treat the control and state estimation problems as decoupled objectives, relying on certainty equivalence controllers that are blind to the varying observability induced by nonlinear open-circuit voltage models. In this paper, we show that for a broad class of objectives, including the peak shaving and valley filling scenarios common in grid-connected energy storage, the expected cost of a stochastic battery system can be exactly parametrized by the conditional mean and covariance of the state of charge. This reformulation reveals a direct coupling between the control input and estimation quality, a coupling that certainty equivalence controllers ignore, and motivates a dual-control approach in which the controller actively reduces estimation uncertainty by driving the state to high observability regions without compromising the control objective. We derive a deterministic surrogate to this stochastic cost and pose the dual-control problem as a computationally tractable model predictive control problem. We validate our approach on a nine-battery system tracking a time-varying reference trajectory. We report simultaneous improvements in tracking cost (a 28\% reduction) and state estimation error (up to 18\% reduction). The estimation improvement is reported across different state estimators: extended Kalman filter, unscented Kalman filter, and a moving horizon estimator, confirming that the estimation improvement of our approach is not restricted to a specific state observer.
\end{abstract}

\section{Introduction}
The US battery capacity increased by $66\%$ in 2024 and continued to grow further in 2025 \cite{EIA_report}. With this rapid growth in the usage of batteries for the power grid, electric vehicles, AI infrastructure and data center projects \cite{Economist_report}, the development of battery management systems for safe and reliable operation is essential \cite{Han2019}. In power grid applications, safe and reliable operation heavily relies on energy management strategies like peak shaving and valley filling, which require batteries to precisely track dynamic state-of-charge (SOC) reference trajectories to balance supply and demand. Executing these strategies hinges on accurate SOC estimation, since the SOC quantifies the charge remaining in the battery. However, SOC cannot be measured directly and must be inferred from indirect observation \cite{ting2014BMS}.

The challenge of SOC estimation for battery management systems invites methodologies from control theory, state estimation, and, in particular, Kalman filtering derivatives and general observer designs. The nonlinearity of the open circuit voltage~(OCV) models necessitates additional considerations beyond standard estimation techniques. The OCV-SOC curves are typically parametrized nonlinear functions fitted empirically \cite{sundaresan2022tabular,pattipati2014open}. Therefore, many nonlinear estimation algorithms and adaptive filters have been developed to interpret the nonlinear behavior of the SOC, including extended Kalman filters (EKFs), Luenberger observers, sliding-mode observers, and proportional-integral observers \cite{he2013adaptive, hu2010estimation, kim2013second, xu2013state}.

A critical but often neglected challenge in OCV-based estimation 
is that the observability of the SOC is not uniform across the 
state-space. OCV-SOC curves for many common battery chemistries 
exhibit voltage plateaus (regions with flat curve where the 
sensitivity of OCV to changes in SOC is 
near zero).\footnote{Observability in our context refers to the 
complete stochastic observability as in \cite{liu2011stochastic}. 
Nearly constant OCV curves admit near-zero output to state
sensitivity, impairing the identifiability and observability 
of the state from the acquired measurements.} In these regions, any state observer, regardless 
of its structure, receives measurements that carry 
little information about the true SOC, leading to degraded 
estimation accuracy and potentially unsafe
operation. This is particularly consequential in peak 
shaving and valley filling, where the SOC is driven across 
a wide range of operating points, passing through both 
high- and low-observability regions of the OCV curve, and 
where estimation errors directly translate into reference 
tracking errors and suboptimal energy dispatch.

Although there is extensive work on improving SOC estimation through more accurate battery models, advanced observers, and more efficient numerical methods, these approaches all operate on the measurements they are given: they cannot recover information that the output does not carry. When the state is poorly observable, the observer design, however sophisticated, cannot accurately infer the SOC. The observability is, however, not fixed and it depends on where the SOC is driven and is therefore shaped by the control input. Certainty equivalence controllers cannot exploit this coupling, since they are blind to the varying observability of the OCV curve.

We address this by posing the problem of battery management system as a stochastic optimal control problem \cite{aastrom2012introduction,feldbaum1960dual}. Rather 
than decoupling the control and estimation objectives, the proposed 
controller optimizes a cost that explicitly incorporates the
estimation quality through the state covariance, 
leading to the emergence of the well-known dual effect 
(caution and probing; see \cite{tse1973wide}) of stochastic 
optimal control. As a consequence, the controller not only 
satisfies the tracking objective but simultaneously steers 
the battery trajectories toward high-observability regions 
of the OCV curve, improving the quality of future 
measurements and reducing estimation error 
\cite{ramadan2024extended}. In particular, we present 
the following contributions:
\begin{itemize}
    \item We show that the problem of battery system management, due to the common nonlinear nature of the OCV functions, can greatly benefit from the stochastic optimal control formulation.
    \item We derive a computationally tractable deterministic surrogate to the stochastic optimal control cost that can be posed as a model predictive control (MPC) problem.
    \item Through a simulation on a nine-battery system tracking a time-varying power-demand reference to imitate peak shaving and valley filling scenarios, we show that the proposed approach achieves up to $28\%$ reduction in tracking cost and $18\%$ reduction in estimation error compared to certainty equivalence control.
    \item We demonstrate filter independence by showing consistent improvement in estimation accuracy across an EKF, an unscented Kalman filter (UKF), and a moving horizon estimator (MHE) while simultaneously reducing the control cost. 
\end{itemize}
 
\section{Battery systems dynamics and observation models of different battery chemistries}
We adopt a zeroth-order equivalent circuit model, neglecting the ohmic voltage drop across the internal resistance. This drop scales with current and is negligible under the moderate charge/discharge rates typical of grid storage.

The SOC propagates according to Coulomb counting, which is the integration of current drawn from or supplied to a battery over time, $\tau \in [t_0, t]$, which gives the state of charge as (for a battery indexed by $i$, $i \in \{1,2,\hdots,n\}$, $n$ being the number of batteries in the system)
\begin{equation}
SOC^i(t) = SOC^i(t_0) - \int_{t_0}^t\dfrac{\eta^i I^i(\tau)}{Q^i_{nom}}d\tau,\label{eq:colCount}
\end{equation}
where $I^i(t)$ is the current, $Q^i_{nom}$ is the nominal battery capacity, and $\eta^i$ is the battery efficiency. The discrete-time analog:
\begin{equation}
    x^i_{k+1}=x^i_k + \frac{\eta^i \Delta t}{Q^i_{nom}}I^i_k + w^i_k\label{eq:discreteColCount}
\end{equation}
where $x^i_k=SOC(t_0 + k \Delta t)$, for some sampling-time $\Delta t$, and $w^i_k$ is an additive process disturbance, representative of uncertainties such as integration errors, parasitic loads, temperature effects, and fluctuations in power supply. 

Due to the stochasticity of the dynamics \eqref{eq:discreteColCount} and the initial condition $x^i_0$, the estimate of the SOC is refined by an OCV-SOC mapping, which is typically a nonlinear curve represented by 
\begin{equation}
    y^i_k = h^i(x^i_k) + v^i_k \label{eq:OCV}
\end{equation}
where $y^i_k$ denotes the open-circuit voltage of battery $i$ and $v^i_k$ is an additive measurement noise. The nonlinear function $h^i(x^i_k)$ can be modeled in many different forms, including a polynomial fit \cite{zhao2017observer}, exponential fit \cite{lazreg2020lithium}, or sum of sines \cite{zhang2018OCV}, depending on the configuration and chemistry of the underlying battery. A variety of OCV-SOC curves can be found in \cite{zhao2017observer}, \cite{zhang2018OCV}, \cite{wang2021unscented}, and \cite{stroe2018LTO}. Lacking an analytic OCV-SOC relationship, these fits rely on experimental data.

Although the OCV-SOC mappings $h^i(x^i_k)$ vary with chemistry, temperature, state of health, and curve fitting techniques, they are typically monotonic and differentiable in the SOC \cite{Pillai2022OCV}. Its plateaus and nonlinearities, however, make observability vary across the state space, motivating the stochastic optimal control approach we develop to improve estimation and control jointly.

\section{Stochastic optimal control and its deterministic surrogate}
We consider a multi-battery system of the following configuration
\begin{subequations}\label{eq:stateSpace}
\begin{align}
x_{k+1}&=f(x_k, I_k) + w_k, \quad w_k \sim \mathcal{N}(0,\Sigma_w),\label{eq:stateDynamics} \\
y_k &= h(x_k) + v_k, \quad v_k \sim \mathcal{N}(0,\Sigma_v),\label{eq:outputDynamics}
\end{align}
\begin{align*}
    x_k &= \left [x^1_k,x^2_k,\hdots x^n_k\right]^T,\quad I_k = \left [I^1_k,I^2_k,\hdots I^n_k)\right]^T\\
    y_k &= \left [y^1_k,y^2_k,\hdots y^n_k\right]^T,\\
    f(x_k, I_k) &= x_k + \text{diag}( \left \{\frac{\eta^i \Delta t}{Q_{nom}^i}:\, i=1,2,\hdots,n \right \}) I_k,\\
    h(x_k) &= \left [h^1(x^1_k),h^2(x^2_k),\hdots h^n(x^n_k)\right]^T,
\end{align*}
\end{subequations}
where the additive noises $w_k\in\mathbb R^{n}$ and $v_k\in\mathbb R^{n}$ are the stacked vectors of $w_k^i$ and $v_k^i$, and are each independent and identically distributed with zero means and diagonal covariance matrices $\Sigma_w \succeq 0$ and $\Sigma_v \succ 0$ (positive semi-definite and positive definite, respectively). They are mutually independent from each other and from $x_0$. As we intend to design a receding-horizon control, the current time-step will be denoted by $k_0$, and the initial (current) state $x_{k_0}$ is assumed to have bounded mean $x_{k_0 \mid k_0}$ and covariance $\Cov(x_{k_0}) = \Sigma_{k_0 \mid k_0}$.

We consider a cost function of the following form (the arguments of $J^N_{k_0}(X_{k_0}, U^{k_0+N}_{k_0})$ are suppressed for compactness)
\begin{align}
    J_{k_0}&= \E \sum_{k=k_0}^{k_0 + N} \Big ( c \left  (q^\top x_k - r_k\right )^2 + I_k ^\top R I_k + \nonumber\\
    &\hskip 5mm c_0 \cdot \sum_{\substack{1 \leq i  \leq n \\ i < j \leq n}} (x_{k}^i-x_{k}^j)^2 \Big), \label{eq:costRaw}
\end{align}
where $U_{k_0}^{k_0+N}=\{I_{k_0},\hdots,I_{k_0+N}\}$, the tuple $X_{k_0}=( x_{k_0 \mid k_0},\Sigma_{k_0 \mid k_0})$, $r_k$ is an SOC reference trajectory, $q=\left [q_1,\hdots,q_n\right ]$ has positive elements representing the nominal energy capacity of each battery, $c,c_0 \geq 0$, and the matrix $R\succ 0$ (positive definite symmetric). The expectation is over all the random variables $x_{k_0}$, $w_k$ and $v_k$, for $k=k_0,k_0+1,\hdots,k_0+N$. The inner double-summation is a similarity/uniformity cost to penalize unequal participation in storage between batteries.

We assume the following sequence of events: at the beginning of each time step $k > k_0$, $y_k$ is observed prior to the execution of the control action $u_k$. That is, the most up-to-date causal control law $u_k$ is a function of the accessible data up to and including time $k$. In particular, $I_k = I_k (Z^{k}_{k_0})$, where $Z^{k}_{k_0}=\{X_{k_0},Y_{k_0}^{k},U_{k_0}^{k-1}\}$, $Y_{k_0}^{k} = \{y_{k_0+1},\hdots,y_{k}\}$, $U_{k_0}^{k-1}= \{I_{k_0},\hdots,I_{k-1}\}$. For notational consistency, we denote the current time information by $Z^{k_0}_{k_0}=\{X_{k_0}\}$.

Note that at each current time $k_0$, all previous input/output data required to characterize the future ($k \geq k_0$) state distribution are assumed to be captured by $X_{k_0}$. That is, the distribution of $x_{k_0}$, given all previous information, is assumed to be captured by a Gaussian of mean $ x_{k_0 \mid k_0}$ and covariance $\Sigma_{k_0 \mid k_0}$. This is a property of Bayesian filtering, assuming that a Gaussian density is sufficiently accurate to capture the distribution of $x_{k_0}$.

We first show that the cost \eqref{eq:costRaw} can be rewritten\footnote{The left expectation is $\E\{\cdot \} = \int \cdot\, p(x_k) dx_k$ with respect to the density function $p(x_k)$, which by the Markov property, $p(x_k)=p(x_{k_0})p(x_{k_0+1}\mid x_{k_0})\hdots p(x_k,x_{k-1})$, and the initial density $p(x_{k_0})$ is characterized by $X_{k_0}$. The right expectation, however, $\E \left \{ \cdot \right \} = \int \cdot\, p(Y^k_{k_0} dY^k_{k_0}.$. That is, we condition on measurements and then marginalize over them as in \eqref{eq:smoothing_theorem}.} in terms of the conditional means and covariances of $x_k$
\begin{lemma} \label{lemma:smoothing_thm}
The term $\E \left  (q^\top x_k - r_k\right )^2$ can be equivalently expressed by
\begin{align}
    \E \left  (q^\top x_k - r_k\right )^2 = \E \left \{ \left  (q^\top x_{k \mid k} - r_k\right )^2 + q^\top \Sigma_{k \mid k} q\right \}, \label{eq:lemma1 result}
\end{align}
where $x_{k \mid k}$ and $\Sigma_{k \mid k}$ are the mean and covariance of $x_k$, conditioned on the information $Z^k_{k_0}$,
\begin{equation}
\begin{aligned}\label{eq:filterMeanCov}
     x_{k \mid k} &= \E \left \{ x_k \mid Z^k_{k_0}\right \},\\
    \Sigma_{k \mid k} &= \E \left \{ [x_k-x_{k \mid k}][x_k-x_{k \mid k}]^\top \mid Z^k_{k_0} \right \}.
\end{aligned}
\end{equation}
(The notation $()_{k \mid j}$, as is typical in state estimation literature \cite[Ch.~3]{anderson2012optimal}, denotes a random variable at time $k$ conditioned on the information $Z^j_{k_0}$, for $j \geq k_0$.)
\end{lemma}
\begin{proof}
By the smoothing theorem \cite[p.~348]{resnick2019probability}, conditioning each additive term in \eqref{eq:costRaw} on the information available at that corresponding time $k$ (respecting causality), we have
\begin{align}
    &\hskip -5mm\E \left \{ \left  (q^\top x_k - r_k\right )^2 \right \} = \int \left  (q^\top x_k - r_k\right )^2 p(x_k)dx_k,\nonumber\\
    &= \int\int \left  (q^\top x_k - r_k\right )^2 p(x_k,Y^k_{k_0})d\,Y^k_{k_0} dx_k,\nonumber\\
    &= \int \left (\int\left  (q^\top x_k - r_k\right )^2 p(x_k \mid Y^k_{k_0}) dx_k \right )p(Y^k_{k_0}) Y^k_{k_0},\nonumber\\
    &=\E \left \{ \E \left \{ \left  (q^\top x_k - r_k\right )^2  \mid Z^k_{k_0}\right\}\right \}.\label{eq:smoothing_theorem}
\end{align}
It is left to show that $\E \left \{\left  (q^\top x_k - r_k\right )^2  \mid Z^k_{k_0}\right\} =  \left  (q^\top x_{k \mid k} - r_k\right )^2 +  q^\top \Sigma_{k \mid k} q$. We first decompose the state $x_k$ into $x_k =  x_{k \mid k} + \tilde x_k$, where $\tilde x_k$ is the estimation error. The quadratic term $\left  (q^\top x_k - r_k\right )^2 = ( \left  (q^\top (x_{k \mid k} + \tilde x_k) - r_k\right )^2$, which under the conditional expectation $\E \left \{ \cdot \mid Z^k_{k_0} \right \}$, and after ignoring the zero mean cross-terms, leads to the right-hand-side of \eqref{eq:lemma1 result}. 
\end{proof}

\begin{lemma} \label{lemma:smoothing_thm similarity cost}
The double-summation $\E  \sum_{\substack{1 \leq i  \leq n \\ i < j \leq n}} (x_{k}^i-x_{k}^j)^2$ can be equivalently expressed by
\begin{align}
    \E \sum_{\substack{1 \leq i  \leq n \\ i < j \leq n}} (x_{k}^i-x_{k}^j)^2 &=  \sum_{\substack{1 \leq i  \leq n \\ i < j \leq n}} \Big \{ (\Sigma_{k \mid k}^{i,i}-2\Sigma_{k \mid k}^{i,j} + \Sigma_{k \mid k}^{j,j}) + \nonumber\\
    &\hskip 15mm (x_{k \mid k}^i-x_{k \mid k}^j)^2 \Big \}, \label{eq:lemma2 result}
\end{align}
where $\Sigma_{k \mid k}^{i,j}$ is the element at the $i^{th}$ row and $j^{th}$ column of the (positive semi-definite symmetric) covariance matrix $\Sigma_{k \mid k}$. Similarly to Lemma~\ref{lemma:smoothing_thm} in \eqref{eq:smoothing_theorem}, the expectation to the left is with respect to $p(x_k)$ and the one on the right is with respect to $p(Y_{k_0}^k)$.
\end{lemma}
\begin{proof}
Analogously to the proof of Lemma~\ref{lemma:smoothing_thm}, applying the smoothing theorem on each element of the summation
\begin{align*}
    \E \left \{ (x_k^i)^2 - 2 x_k^i x_k^j + (x_k^j)^2) \right\} &= \\
    &\hskip -20mm \E \left \{\E \left \{ (x_k^i)^2 - 2 x_k^i x_k^j + (x_k^j)^2) \mid Z_{k_0}^k \right\} \right \},
\end{align*}
then using the decomposition $x^i_k =  x^i_{k \mid k} + \tilde x_k^i$, we can rewrite the inner conditional expectation in terms of the elements of $x_{k \mid k}$ and $\Sigma_{k \mid k}$.
\end{proof}
\begin{corollary}
The cost function \eqref{eq:costRaw} can be equivalently represented by 
\begin{align}
& J_{k_0} = \nonumber \\
     &\E \Big \{ \sum_{k=k_0}^{k_0 + N} \Big [c \left  (q^\top x_{k \mid k} - r_k\right )^2 + c \cdot q^\top \Sigma_{k \mid k} q) + I_k ^\top R I_k \nonumber + \\
     &c_0\sum_{\substack{1 \leq i  \leq n \\ i < j \leq n}} \big \{ \Sigma_{k \mid k}^{i,i}-2\Sigma_{k \mid k}^{i,j} + \Sigma_{k \mid k}^{j,j} + (x_{k \mid k}^i-x_{k \mid k}^j)^2 \big \} \Big ]\Big \}. \label{eq:costWideSense}
\end{align}
\end{corollary}
\begin{proof}
A direct consequence of applying Lemmas~\ref{lemma:smoothing_thm} and \ref{lemma:smoothing_thm similarity cost} to each additive term in \eqref{eq:costRaw}.
\end{proof}

\begin{remark}
It is important to emphasize here that the cost parametrization in terms of the lower order statistics, $x_{k \mid k}$ and $\Sigma_{k \mid k}$, in \eqref{eq:costWideSense}, is exactly equivalent to the original cost in \eqref{eq:costRaw}, with no approximations. We employ the first approximation to $J_{k_0}$ in the next subsection.
\end{remark}

\subsection{Mean and covariance dynamics}
We shifted from a cost in the partially observed state $x_k$ in \eqref{eq:costRaw} into a cost in the fully observed 'information state', represented by the pair $(x_{k \mid k}, \Sigma_{k \mid k})$. The latter description is fully observed as it can be tracked and calculated from a Bayesian filter. A natural choice here is the EKF, as it is one of the minimalist algorithms to approximate the evolution of the first two central moments.

In addition to \eqref{eq:filterMeanCov}, we define
\begin{equation}
\begin{aligned} \label{eq:one step ahead}
     x_{k+1 \mid k} &= \E \left \{ x_{k+1} \mid Z^k_{k_0},u_k \right \}, \\
    \Sigma_{k+1 \mid k} &= \E \left \{ [x_{k+1}-x_{k+1 \mid k}][x_{k+1}-x_{k+1 \mid k}]^\top \mid Z^k_{k_0},u_k \right \}.
\end{aligned}  
\end{equation}
The EKF recursion is defined by
\begin{align}
 x_{k+1 \mid k} = f( x_{k \mid k},I_k),\quad \label{eq:EKF_state}
\Sigma_{k+1 \mid k} = F_k \Sigma_{k \mid k} F_k^\top + \Sigma_w, 
\end{align}
where
\begin{equation}
\begin{aligned} \label{eq:EKF_supportVariables}
& x_{k\mid k}= x_{k \mid k-1} + \Omega_{k} \left [y_{k}-h(x_{k \mid k-1})\right],\\
&\Sigma_{k \mid k} = \left [I - \Omega_{k} H_{k} \right]\Sigma_{k \mid k-1},\\
 &\Omega_{k}= \Sigma_{k \mid k-1} H_{k}^\top \left [ H_{k} \Sigma_{k \mid k-1} H_{k}^\top+\Sigma_v\right]^{-1}, \\
 &F_k = \left. \frac{\partial f(x,I_k)}{\partial x} \right | _{x_{k\mid k}} = \mathbb I_{n \times n},\\
 &H_{k} = \left. \frac{\partial h(x)}{\partial x} \right | _{x_{k\mid k-1}},\\
 &\hskip5mm = \text{diag}( \left \{\frac{\partial h^i(x)}{\partial x} | _{x^i_{k\mid k-1}}:\, i=1,\hdots,n \right \}),
\end{aligned}
\end{equation}
and the recursion is initialized by $ x_{k_0 \mid k_0}$, $\Sigma_{k_0 \mid k_0}$, given at the beginning of each time-step.

It is important to note that the EKF results from applying first-order approximations (the Jacobians $F_k$ and $H_k$) on nonlinear systems to reach a similar form to the Kalman filter, with $\Omega_k$ mimicking the Kalman gain \cite{jazwinski2007stochastic}. Hence, although the cost description in \eqref{eq:costWideSense} is exact, the central moments provided by the EKF are only approximates and not exact.

The variables $x_{k \mid k},\Sigma_{k \mid k}$ are random since they are functions of the random observation sequence $Y^k_{k_0}$ (not yet available at the current time $k_0$). The expectation in \eqref{eq:costWideSense} averages over all possible future observations, given by the random variable $Y^k_{k_0}$. Next, we employ an important approximation to omit this expectation, and work with a deterministic surrogate instead.

\subsection{Deterministic surrogate to the cost} \label{subsection:cost approx}
For linear systems, the state covariance propagated by the Kalman filter is independent of the control input; hence, the cost is replaced by its deterministic equivalent in the state mean. This is known as the separation principle \cite{aastrom2012introduction}, which states that optimal control can be separated into deterministic optimal control and optimal filtering. However, this does not hold for general nonlinear systems; notice that the covariance estimates provided by the EKF are dependent on the control input. Therefore, the control input has a role in the estimation quality represented by these covariances. The difficulty remains in the expectation of the cost \eqref{eq:costWideSense} resulting from the conditioning on all future measurements. We omit this expectation for computational purposes and verify in the numerical example that the resulting deterministic surrogate is still beneficial to both control and estimation.

The measurement correction term $[y_{k} - h(x_{k\mid k-1})]$, which resembles the innovation sequence in the filtering context, can be assumed white and independent from $\Omega_{k}$, if the EKF is a sufficiently accurate filter (this is true in various applications (see \cite[Sec.~8.2]{anderson2012optimal}). We define the prediction-only recursion 
\begin{align} \label{eq:CE_step}
     x^p_{k+1} =f( x^p_{k },I_k),
\end{align}
where the state $x^p_{k}$ is the surrogate of both $ x_{k\mid k-1}$ and $x_{k \mid k}$ (since the measurement correction is omitted) that evolves solely through prediction. Although $x^p_k$ does not exactly represent the mean of $x_k$, beyond one step ahead, due to the nonlinearity of $f$; we adopt $x_k^p$ as the approximate mean for the deterministic surrogate of the cost.

We define the prediction only EKF as follows.
\begin{equation}
    \begin{aligned} \label{eq:prediction EKF}
    &x_{k+1}^p = x_k^p + \text{diag}( \left \{\frac{\eta^i \Delta t}{Q_{nom}^i}:\, i=1,2,\hdots,n \right \}) I_k,\\
&\Sigma_{k+1 \mid k}^p = \Sigma_{k \mid k}^p + \Sigma_w,\\
&\Sigma_{k \mid k}^p = \left [I - \Omega_{k} H^p_{k} \right]\Sigma_{k \mid k-1}^p,\\
&\Omega_{k}= \Sigma_{k \mid k-1}^p H_{k}^{p\,\,\top} \left [ H^p_{k} \Sigma_{k \mid k-1}^p H_{k}^{p\,\,\top}+\Sigma_v\right]^{-1},\\
&H^p_{k} = \text{diag}( \left \{\frac{\partial h^i(x)}{\partial x} | _{x^{i,p}_{k}}:\, i=1,\hdots,n \right \}),
    \end{aligned}
\end{equation}
with $F^p_k$ substituted as the identity matrix.

We define the deterministic surrogate to \eqref{eq:costWideSense} in terms of the prediction-only EKF.
\begin{align}
      &\hat J_{k_0} = \sum_{k=k_0}^{k_0 + N} \Big [c \left  (q^\top x_{k}^p - r_k\right )^2 + c \cdot q^\top \Sigma_{k \mid k}^p q + I_k ^\top R I_k \nonumber + \\
     &c_0\sum_{\substack{1 \leq i  \leq n \\ i < j \leq n}} \big \{ \Sigma_{k \mid k}^{p\,\,i,i}-2\Sigma_{k \mid k}^{p\,\,i,j} + \Sigma_{k \mid k}^{p\,\,j,j} + (x_{k}^{p\,\,i}-x_{k}^{p\,\,j})^2 \big \} \Big ].\label{eq:costWideSenseDeterministic}
\end{align}

\subsection{Algorithm}
The objective \eqref{eq:costWideSenseDeterministic} is quadratic in $x_k^p$ and (linear in) $\Sigma_{k \mid k}^p$, but the dynamic constraint \eqref{eq:prediction EKF} is generally nonconvex. This makes the problem belong to the realm of nonlinear optimization, in particular, a nonlinear model predictive control (MPC) problem. 

We first list the certainty equivalence linear MPC problem that omits the OCV models.

\begin{problem}\label{prob:linearMPC}
Linear MPC in the SOC dynamics (certainty equivalence, no OCV model)
\begin{align}
    &\min_{U_{k_0}^{k_0+N}, \{x_k\}_{k=k_0}^{k_0 + N} }  \E \sum_{k=k_0}^{k_0 + N} \left ( c \left  (q^\top x_k - r_k\right )^2 + I_k ^\top R I_k + J_{uni} \right )\nonumber\\
    &\text{subject to:}\nonumber\\
    &x_{k+1} = x_k + \text{diag}( \left \{\frac{\eta^i \Delta t}{Q_{nom}^i}:\, i=1,2,\hdots,n \right \}) I_k,\nonumber\\
    &I_k^i \in [I^i_{min}, I^i_{max}], \quad x_k^i \in [0,1],\label{eq:constraint I x}\\
    &i=1,2,\hdots,n,\quad k=k_0,k_0+1,\hdots,k_0+N, \nonumber
\end{align}

where $I^i_{min}$ and $I^i_{max}$ are the minimum and maximum allowable currents, and $J_{uni}$ refers to the uniformity cost.
\end{problem}

\begin{problem}\label{prob:nonlinearMPC}
Nonlinear MPC in the information state $(x_k^p, \Sigma^p_{k \mid k})$
\begin{align*}
    &\min_{U_{k_0}^{k_0+N}, \{(x_k^p, \Sigma^p_{k \mid k})\}_{k=k_0}^{k_0 + N} }  \hat J_{k_0}\\
    &\text{subject to \eqref{eq:prediction EKF} and \eqref{eq:constraint I x}.}
\end{align*}
\end{problem}

Problem~\ref{prob:nonlinearMPC} is nonlinear and nonconvex and can be approached by some nonlinear solver. For simplicity and computational efficiency, we rely on Algorithm~\ref{algorithm:algorithm}, which constructs a linear approximation of the information state dynamics instead of solving the full nonlinear MPC problem. The detailed implementation can be found in our Julia code at \href{https://github.com/msramada/battery-dual-control}{\textit{\small https://github.com/msramada/battery-dual-control}.}

\begin{algorithm}
\small
\caption{Local linearization of Problem~\ref{prob:nonlinearMPC}}\label{algorithm:algorithm}
\begin{algorithmic}[0]
\State At current time $k_0$, given $X_{k_0}=(x_{k_0 \mid k_0},\Sigma_{k_0 \mid k_0})$ from the running EKF;
\State Set $(x^p_{k_0}, \Sigma^p_{k_0 \mid k_0}) \gets (x_{k_0 \mid k_0},\Sigma_{k_0 \mid k_0})$
\State Linearize \eqref{eq:prediction EKF} about
\begin{align*}
    \xi_{k_0} = \begin{bmatrix}
        x^p_{k_0}\\
        \text{up\_vec}(\Sigma^p_{k_0 \mid k_0})
    \end{bmatrix}, \quad I_{k_0} = 0,
\end{align*}
($\text{up\_vec}$ is the vectorized upper-triangle; to avoid repetition due to symmetry) to find $\mathcal{A}, \mathcal{B}$ in
\begin{align} \label{eq:linearized info state dynamics}
    \xi_{k+1} = \mathcal{A} \xi_k + \mathcal{B} I_k;
\end{align}
\State Solve:
\begin{align*}
    &\min_{U_{k_0}^{k_0+N}, \{x_k\}_{k=k_0}^{k_0 + N} }  \hat J_{k_0}\\
    &\text{subject to \eqref{eq:constraint I x}, \eqref{eq:linearized info state dynamics} and}\\
    &\Sigma^p_{k \mid k} \succeq 0, \quad k=k_0,\hdots,k_0+N;
\end{align*}
\State Apply the optimal value of $I_{k_0}$.
\end{algorithmic}
\end{algorithm}

\section{Simulation}

\begin{figure}
    \centering
    \includegraphics[width=0.8\linewidth]{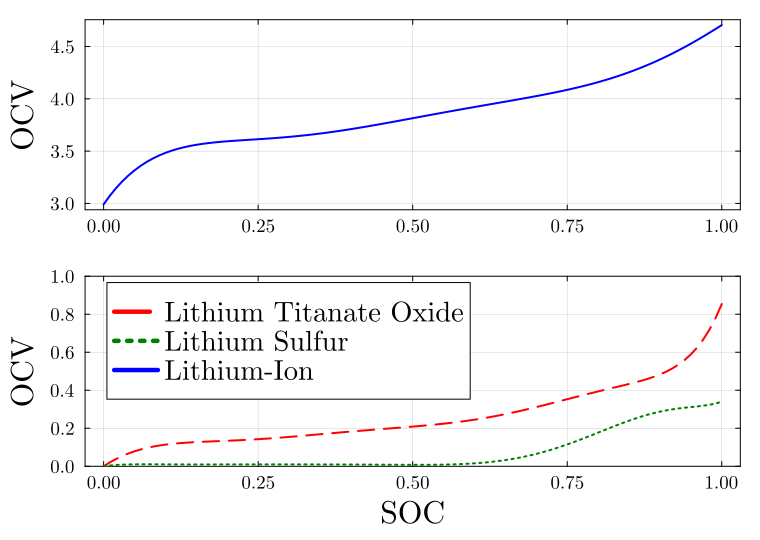}
    \caption{OCV-SOC curves of 3 different battery chemistries}
    \label{fig:ocv_numerical}
\end{figure} 

We test Algorithm~\ref{algorithm:algorithm} in simulation for $n=9$ batteries. The primary goal is to track a dynamic reference trajectory $r_k$ to meet time-varying power demands dictated by peak shaving and valley filling strategies\footnote{We use this SOC-tracking formulation for simplicity, though in practice peak shaving/valley filling is a power-tracking problem involving a bilinear power relation in $I=P/V$.}. This simulation uses 3 of each battery OCV-SOC model shown in Figure \ref{fig:ocv_numerical} for a total of 9-battery system. The measurement dynamics, which are plotted in Figure \ref{fig:ocv_numerical}, are drawn from polynomial-fitted OCV-SOC curves previously established in SOC literature (\cite{wang2021unscented, stroe2018LTO,PROPP2017254}), and the exact polynomial coefficients are listed in Table \ref{tab:sim_params}, together with other system and simulation parameters ($\textbf{1}_n$ is a vector of ones of length $n$).

\begin{figure*}
    \centering
    \includegraphics[width=0.9\linewidth]{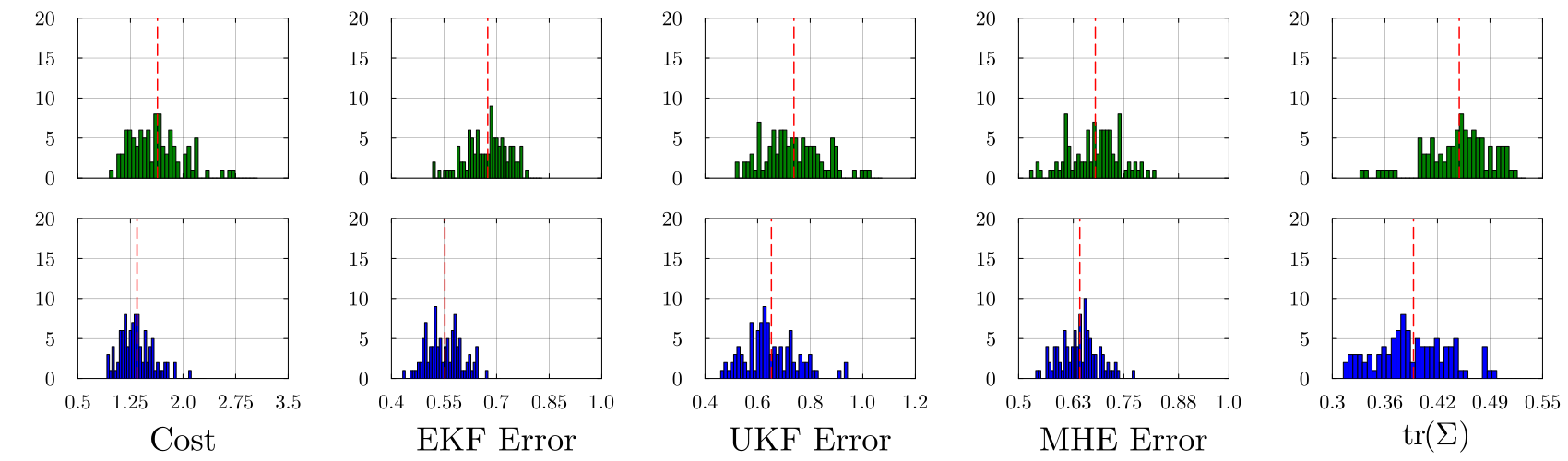}
    \caption{Histograms comparing the sample-average estimation error and cost of 100 simulations. The top row in green represents Linear MPC, the bottom row in blue represents Algorithm \ref{algorithm:algorithm}, and the dashed red line marks the average of the corresponding histogram}
    \label{fig:histograms}
\end{figure*}

\begin{table}[htpb]
    \centering
    \caption{Simulation Parameters }
    \label{tab:sim_params}
    \begin{tabular}{ l c | @{\hskip 0.5cm} l c }
        \toprule
        \textbf{Parameter} & \textbf{Value} & \textbf{Parameter} & \textbf{Value} \\
        \midrule
        $(x_{0|0}, \Sigma_{0|0})$  & $(0.1\textbf{1}_n,0.1\mathbb{I}_{n \times n})$   & $c_0$          & 0.1 \\
        $\eta^i$          & 1.0                    & $q$            & $\textbf{1}_n$  \\
        $c$             & 1.0                    & $R$            & $0.1\mathbb{I}_{n \times n}$   \\
        $\Delta t$      & 1.0                    & $N$            & 8   \\
        $Q_{nom}^i $           & 2.2                    & $T$            & 50 \\
        $(\Sigma_w, \Sigma_v)$                &$(0.02\mathbb{I}_{n \times n}, 0.02\mathbb{I}_{n \times n})$                      & $I_{min}$      & -0.5 \\
        $r_k$                &    $5.0+\sin\big(({2\pi}/{T})k\big)$                    & $I_{max}$      & 0.5 \\
        \midrule
        \multicolumn{4}{l}{\textbf{OCV-SOC models:} coefficients of $(\text{SOC}^0, \text{SOC}^1, \text{SOC}^2, \hdots)$,} \\
        \multicolumn{4}{l}{\footnotesize Li-ion: $(3,\; 8.674,\; -50.16,\; 146.6,-215.7,\; 155.4,\; -43.1)$} \\
        \multicolumn{4}{l}{\footnotesize LTO: $(0,\; 2.1737,\; -13.65,\; 31.77,\; 39.06,$} \\
        \multicolumn{4}{l}{\footnotesize \hspace{2.2em}$-329.2,\; 632.1,\; -526.2,\; 164.8)$} \\
        \multicolumn{4}{l}{\footnotesize Li-S: $(0,\; 0.48,\; -8.36,\; 72.94,\; -364.76,$} \\
        \multicolumn{4}{l}{\footnotesize \hspace{2.2em}$1107.76,\; -2066.02,\; 2291.23,\; -1372.71,\; 339.78)$} \\
        \bottomrule
    \end{tabular}
\end{table}

To evaluate the performance of our control approach, 100 simulations of this numerical example were run using the certainty-equivalence linear MPC approach described in Problem 1, and another 100 simulations using Algorithm 1. For each of these 200 trajectories, the sample cost is calculated by substitution in \eqref{eq:costRaw}, and the estimation error of different filters is evaluated (since the true hidden state is known in simulation). Histograms of these sample costs and sample estimation errors are reported in Figure \ref{fig:histograms}. Note that Algorithm~1 shifted the distributions to the better, decreasing the average cost by $17.9\%$ and decreasing the average EKF estimation error by $18.2\%$. Note also the improvement in estimation accuracy for other filters like the UKF and MHE, since our controller improves observability, which is universally important for any state estimation algorithm. 

\begin{figure*}
    \centering
    \includegraphics[width=0.8\linewidth]{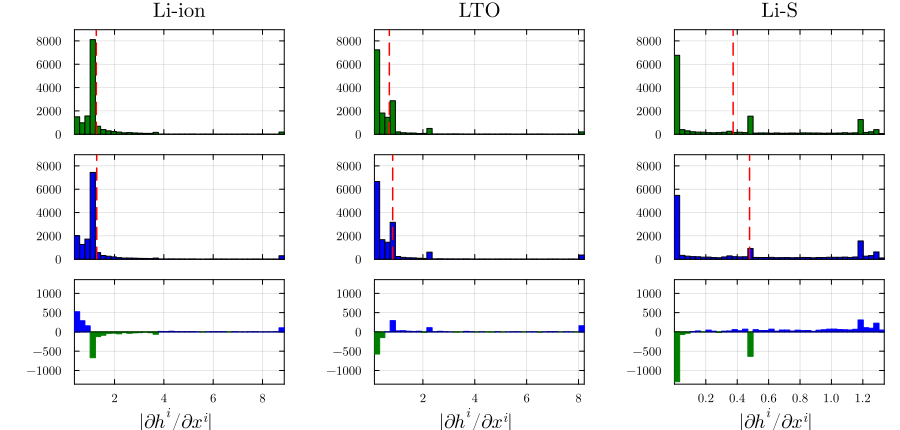}
    \caption{Histograms of $|\partial h^i/\partial x^i|$, an observability proxy: the absolute value of the slope of the OCV-SOC curve evaluated along the realized (true) SOC trajectories. Counts are pooled over $100$ simulations with $T=50$ time steps each, aggregated over each of the $3$ battery chemistries ($100 \times 50 \times 3 = 15{,}000$ samples per panel). Columns correspond to the three chemistries of Figure~\ref{fig:ocv_numerical}; top row represents Linear MPC (Problem~\ref{prob:linearMPC}), second row represents Algorithm~\ref{algorithm:algorithm}, dashed red vertical line marks the mean, and third row is the difference (second row minus first row).}
    \label{fig:observability_histograms}
\end{figure*}

To examine the mechanism behind these improvements in estimation quality, Figure~\ref{fig:observability_histograms} reports how frequently each level of the observability proxy $|\partial h^i/\partial x^i|$ is visited under each controller. Since a near-zero slope of the OCV-SOC curve renders the SOC nearly unobservable from the voltage measurement, a distribution concentrated at larger values indicates more frequent operation in more observable regions. The effect is most pronounced for the Li-S cells, whose OCV-SOC curve carries the widest plateau region: Algorithm~\ref{algorithm:algorithm} moves the mean slope from $0.374$ to $0.479$ (a $28.0\%$ increase), while the median rises from $0.106$ to $0.357$, a factor of $3.4$. The same shift is visible but weaker for the LTO cells ($0.706$ to $0.834$, an $18.2\%$ increase in the mean slope), and is weakest for the Li-ion cells ($1.261$ to $1.278$, a $1.4\%$ increase), whose curve is already well-conditioned over most of its range. In other words, the dual effect is activated selectively, in the chemistries and operating regions where the information gain is worth the tracking compromise.

Furthermore, examining the trajectories of the battery system in Figure \ref{fig:soc_time_domain}, not only do the batteries under Algorithm \ref{algorithm:algorithm} track the reference trajectory more accurately than Linear MPC, but the $2\sigma$ bounds of the true (hidden, but accessible in simulation) state become tighter around $r_k$ as well.

\begin{table*}[htpb] 
    \centering
    \caption{Average achieved cost, state estimation error, and covariance trace between Algorithm \ref{algorithm:algorithm} and certainty-equivalence MPC.}
    \label{tab:average_change}
    \begin{tabular}{l c c c c c c}
        \toprule
        \textbf{Method} & \textbf{Total Cost} & \textbf{Tracking Cost} & \textbf{EKF Error} & \textbf{UKF Error} & \textbf{MHE Error} & \textbf{Cov Trace} \\ 
        \midrule
        Linear MPC & {1.636} & {1.146} & {0.675} & {0.738} & {0.684} & {0.741} \\ 
        \midrule Algorithm~\ref{algorithm:algorithm} & {1.343 (\textbf{-17.9\%})} & {0.826 (\textbf{-28.0\%})} & {0.552 (\textbf{-18.2\%})} & {0.653 (\textbf{-11.6\%})} & {0.646 (\textbf{-5.6\%})} & {0.397 (\textbf{-46.4\%})} \\ 
        \bottomrule
    \end{tabular}
\end{table*}

Notably, the trade-off between estimation accuracy and tracking performance is affected by the reference trajectory $r_k$. Higher "mean" of $r_k$ reduce the redundancy available for charge distribution and limit the algorithm's ability to exploit observable regions. These behaviors indicate that the stochastic optimal control algorithm is most effective when the battery system possesses sufficient resource redundancy, providing a rich feasible set from which the optimization can derive solutions that outperform standard linear MPC.

\begin{table*}[htpb]
    \centering
    \caption{Effect of the battery-pack composition on the performance of Algorithm~\ref{algorithm:algorithm} relative to certainty-equivalence linear MPC. Each composition is written as the number of Li-ion / LTO / Li-S cells in the nine-battery system. Every entry is averaged over $100$ simulations, and the percentage in parentheses is the change with respect to the linear MPC of the same composition. Negative values indicate improvement.}
    \label{tab:composition}
    \begin{tabular}{l l c c c c c c}
        \toprule
        \textbf{Composition} & \textbf{Method} & \textbf{Total Cost} & \textbf{Tracking Cost} & \textbf{EKF Error} & \textbf{UKF Error} & \textbf{MHE Error} & \textbf{Cov Trace} \\
        \midrule
        \multirow{2}{*}{3 / 3 / 3} & Linear MPC & 1.636 & 1.146 & 0.675 & 0.738 & 0.684 & 0.741 \\
         & Algorithm~\ref{algorithm:algorithm} & 1.343 (\textbf{-17.9\%}) & 0.826 (\textbf{-28.0\%}) & 0.552 (\textbf{-18.2\%}) & 0.653 (\textbf{-11.6\%}) & 0.646 (\textbf{-5.6\%}) & 0.397 (\textbf{-46.4\%}) \\
        \midrule
        \multirow{2}{*}{9 / 0 / 0} & Linear MPC & 0.917 & 0.614 & 0.323 & 0.344 & 0.350 & 0.124 \\
         & Algorithm~\ref{algorithm:algorithm} & 0.930 (\textbf{+1.4\%}) & 0.607 (\textbf{-1.2\%}) & 0.322 (\textbf{-0.3\%}) & 0.340 (\textbf{-1.1\%}) & 0.355 (\textbf{+1.3\%}) & 0.125 (\textbf{+0.2\%}) \\
        \midrule
        \multirow{2}{*}{0 / 9 / 0} & Linear MPC & 1.451 & 1.022 & 0.562 & 0.555 & 0.614 & 0.412 \\
         & Algorithm~\ref{algorithm:algorithm} & 1.365 (\textbf{-5.9\%}) & 0.858 (\textbf{-16.0\%}) & 0.534 (\textbf{-5.0\%}) & 0.542 (\textbf{-2.5\%}) & 0.616 (\textbf{+0.3\%}) & 0.323 (\textbf{-21.5\%}) \\
        \midrule
        \multirow{2}{*}{0 / 0 / 9} & Linear MPC & 3.124 & 2.438 & 0.993 & 1.125 & 0.975 & 1.716 \\
         & Algorithm~\ref{algorithm:algorithm} & 1.927 (\textbf{-38.3\%}) & 1.186 (\textbf{-51.3\%}) & 0.829 (\textbf{-16.5\%}) & 0.959 (\textbf{-14.7\%}) & 0.909 (\textbf{-6.8\%}) & 1.092 (\textbf{-36.4\%}) \\
        \bottomrule
    \end{tabular}
\end{table*}

To assess how the reported gains depend on the makeup of the pack, we repeat the study for three homogeneous packs in addition to the balanced $3/3/3$ configuration. The results are collected in Table~\ref{tab:composition}. For the homogeneous Li-ion pack, whose OCV-SOC curve is well-conditioned (in its observability distribution) almost everywhere, every metric lies within roughly one percent of the certainty-equivalence baseline, so the proposed approach neither helps nor meaningfully hurts. The homogeneous LTO pack, with a milder plateau, shows an intermediate benefit, while the homogeneous Li-S pack, which is the stress test of interest since it is the least observable of the three chemistries, yields by far the largest gains.

\begin{table*}[htpb]
    \centering
    \caption{Parameter sensitivity for the balanced $3/3/3$ composition. One parameter is varied at a time while all others remain at their Table~\ref{tab:sim_params} values. Every entry is averaged over $100$ simulations. Negative values indicate improvement.}
    \label{tab:sensitivity}
    \begin{tabular}{l c c c @{\hskip 0.6cm} c c c}
        \toprule
        & \multicolumn{3}{c}{\textbf{Tracking Cost}} & \multicolumn{3}{c}{\textbf{EKF Error}} \\
        \cmidrule(lr){2-4} \cmidrule(lr){5-7}
        \textbf{Case} & Linear MPC & Algorithm~\ref{algorithm:algorithm} & \textbf{Change} & Linear MPC & Algorithm~\ref{algorithm:algorithm} & \textbf{Change} \\
        \midrule
        Nominal ($c_0 = 0.1$, $\Sigma_w =\Sigma_v = 0.02\,\mathbb{I}$, $N=8$) & 1.146 & 0.826 & \textbf{-28.0\%} & 0.675 & 0.552 & \textbf{-18.2\%} \\
        $c_0 = 0.01$ ($10\%$ of nominal) & 1.052 & 0.679 & \textbf{-35.4\%} & 0.650 & 0.477 & \textbf{-26.6\%} \\
        $\Sigma_w = 0.002\,\mathbb{I}$ ($10\%$ of nominal) & 0.577 & 0.500 & \textbf{-13.4\%} & 0.427 & 0.376 & \textbf{-12.0\%} \\
        $\Sigma_v = 0.002\,\mathbb{I}$ ($10\%$ of nominal) & 1.038 & 0.749 & \textbf{-27.9\%} & 0.644 & 0.448 & \textbf{-30.4\%} \\
        $N = 16$ (double the horizon) & 1.146 & 0.799 & \textbf{-30.3\%} & 0.675 & 0.530 & \textbf{-21.4\%} \\
        \bottomrule
    \end{tabular}
\end{table*}

Table~\ref{tab:sensitivity} varies the uniformity weight, the two noise covariances and the horizon length, one at a time, about the nominal setting. Regarding the uniformity term, penalizing deviations between the individual SOC levels is in conflict with an objective that drives the batteries toward the high-slope regions of their respective OCV-SOC curves, since those regions are located at different SOC values for different chemistries. Reducing $c_0$ relaxes the balancing requirement and grants the controller more freedom to spread the individual SOC levels. Regarding the noise levels, we note first that the nominal values are conservative rather than optimistic. Reducing either covariance by a factor of ten brings the problem to more realistic noise levels. Both controllers improve in absolute terms under either reduction, and the advantage of Algorithm~\ref{algorithm:algorithm} persists throughout. Regarding doubling the horizon length, this extra lookahead buys nothing for the certainty-equivalence linear MPC problem, since $N=8$ can be enough for convergence. On the other hand, for Algorithm~\ref{algorithm:algorithm}, this extends the window over which an information-seeking action is applied, and hence, produces a further improvement.

The enhanced control and estimation provided by Algorithm \ref{algorithm:algorithm} incurs only a slight computational cost compared to standard Linear Quadratic Regulator control. Additionally, the hours-long time horizons and slow control frequencies that are common for battery system problems allow easy scaling of the number of batteries in the system while still maintaining a sufficient speed for most BMS purposes. To quantify this comparison, we timed both controllers on the same machine: Apple M1 Max with 64 GB of RAM, measuring each solve over $300$ closed-loop states sampled from the recorded trajectories. Problem~\ref{prob:linearMPC} solves in a median of $4.1$~ms and Algorithm~\ref{algorithm:algorithm} in a median of $10.7$~ms. Both remain orders of magnitude faster than the sampling rates required in battery management applications. The two problems are moreover of the same class; since $\Sigma_{k \mid k}$ is diagonal, both of these problems are quadratic programs. Using a primal-dual interior point solver, both have $O( n^{3})$ computational complexity ($n$ is the number of batteries).

Algorithm \ref{algorithm:algorithm} maximizes the performance of available resources according to the cost function, which namely promotes improved reference tracking and state estimation; however, sustained operations at higher SOCs may lead to battery degradation or unequal battery usage in real world systems. While our approach does not explicitly address resource management considerations, the availability or balanced usage of such resources may be supervised by other higher level planners or included in additional cost function terms.

\begin{figure}
    \centering
    \includegraphics[width=0.9\linewidth]{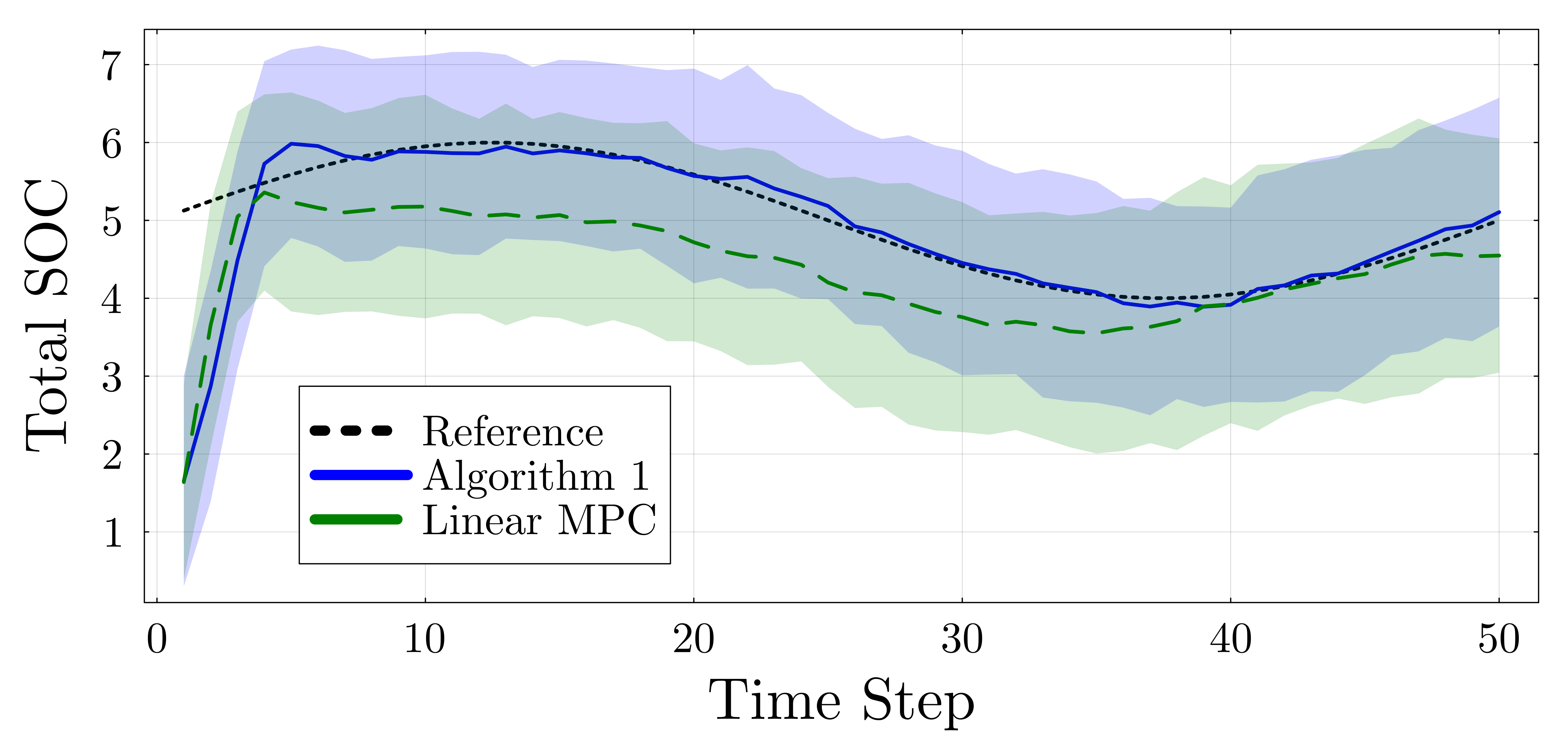}
    \caption{Average trajectories of the true total SOC of the battery system for the linear MPC and Algorithm~1, with the $2\sigma$ regions around each average (purple: Algorithm~1, green: Linear MPC, and grey is the overlap of these two regions). }
    \label{fig:soc_time_domain}
\end{figure}





\section{Conclusion} 
In this work, we employ stochastic optimal control on battery management systems to achieve a simultaneous improvement in control and estimation. This is achieved by deriving a deterministic surrogate to the stochastic optimal control cost, parametrized by the mean and covariance of the state of charge, which exposes a direct coupling between the control input and estimation quality that certainty-equivalence controllers ignore. The controller actively steers battery trajectories toward high-observability regions of the OCV curve, reducing estimation uncertainty while simultaneously improving tracking performance. Although our parameterization is independent of the type of the state observer, we adopt the extended Kalman filter to approximate the information-state dynamics, due to its simplicity. We validate the approach on a nine-battery system tracking a time-varying reference and report significant reductions in tracking cost and estimation error relative to certainty equivalence MPC. These improvements are consistent across an EKF, a UKF, and an MHE, confirming that the dual-control benefit is agnostic to the choice of the state observer. 

Future work involves the adaptation of other Bayesian filters and estimation algorithms into the proposed stochastic optimal control technique, as well as extending the framework to more sophisticated battery models and grid-level objectives such as frequency regulation and power dispatch.

\bibliographystyle{IEEEtran}
\balance
\bibliography{references}

\vspace{0.1cm}
\begin{flushright}
	\scriptsize \framebox{\parbox{2.5in}{Government License: The
			submitted manuscript has been created by UChicago Argonne,
			LLC, Operator of Argonne National Laboratory (``Argonne").
			Argonne, a U.S. Department of Energy Office of Science
			laboratory, is operated under Contract
			No. DE-AC02-06CH11357.  The U.S. Government retains for
			itself, and others acting on its behalf, a paid-up
			nonexclusive, irrevocable worldwide license in said
			article to reproduce, prepare derivative works, distribute
			copies to the public, and perform publicly and display
			publicly, by or on behalf of the Government. The Department of Energy will provide public access to these results of federally sponsored research in accordance with the DOE Public Access Plan. http://energy.gov/downloads/doe-public-access-plan. }}
	\normalsize
\end{flushright}

\end{document}